\documentclass[12pt]{amsart}
\usepackage{amsmath,amssymb,amsthm,mathrsfs,eucal}
\usepackage{graphicx}
\usepackage[american]{babel}
\usepackage{a4wide}

\newtheorem{theorem}{THEOREM}
\newtheorem{assumption}{ASSUMPTION}

\usepackage{enumitem}
\newenvironment{enumproof}{
	
	\begin{enumerate}[leftmargin=0cm,itemindent=1cm,itemsep=2mm]
	}
{ 	\end{enumerate} }

\newcommand{\lint}{\int\limits}
\newcommand{\Lip}[1]{\textnormal{Lip}(#1)}
\newcommand{\R}{\mathbb{R}}
\renewcommand{\u}{\mathbf{u}}
\renewcommand{\v}{\mathbf{v}}
\newcommand{\T}{\mathbb{T}}
\newcommand{\bphi}{\boldsymbol{\phi}}
\newcommand{\bPhi}{\boldsymbol{\Phi}}
\newcommand{\bPsi}{\boldsymbol{\Psi}}
\newcommand{\A}{\CMcal{A}}

\graphicspath{{./}{figures/}}

\title{Multiphase modeling of tumor growth with matrix remodeling and fibrosis}
\author{Andrea Tosin}
\address{Department of Mathematics \\
		Politecnico di Torino \\
		Corso Duca degli Abruzzi 24, 10129 Torino, Italy}
\email{andrea.tosin@polito.it}
\thanks{A. Tosin acknowledges the support of a fellowship by the National Institute for Advanced Mathematics (INdAM) and the ``Compagnia di San Paolo'' foundation.}
\author{Luigi Preziosi}
\address{Department of Mathematics \\
		Politecnico di Torino \\
		Corso Duca degli Abruzzi 24, 10129 Torino, Italy}
\email{luigi.preziosi@polito.it}

\begin{document}

\begin{abstract}
We present a multiphase mathematical model for tumor growth which incorporates the remodeling of the extracellular matrix and describes the formation of fibrotic tissue by tumor cells. We also detail a full qualitative analysis of the spatially homogeneous problem, and study the equilibria of the system in order to characterize the conditions under which fibrosis may occur.
\end{abstract}


\maketitle

\section{Introduction}
It is well known that tumor tissues are often stiffer than normal tissues. For instance, a normal mammary gland has an elastic modulus of about $2$ hPa, which may dramatically increase for a breast tumor up to about $4$ kPa \cite{paszek2005tensional}. For this reason self-palpation is often a successful tool of pre-diagnosis for the detection of possible stiffer nodules, therefore so encouraged. In most cases, the increased stiffness is due to the presence of a denser and more fibrous stroma \cite{butcher2009tense,kass2007mammary,takeuchi1976variation} coming from a considerable change in the content of ExtraCellular Matrix (ECM). Indeed, as reported in \cite{paszek2005tensional}, doubling the percent amount of collagen would increase the stiffness of a tissue by almost one order of magnitude ($328$ Pa and $1589$ Pa for a $2$ mg/ml and a $4$ mg/ml collagen mixture, respectively). The percentage of ECM also changes within the same tumor type during tumor progression \cite{zhang2003characteristics}.

The continuous remodeling of ECM is a physiologically functional process, because it allows to keep the stroma young and reactive. In fact, prolonged rest is detrimental for bones and muscles, while physical training has an opposite effect. The ECM is constantly renewed through the concomitant production of Matrix
MetalloProteinases (MMPs) and new ECM components. In stationary conditions, remodeling of ECM is a slow process: for instance, in human lungs the physiological turnover of ECM is $10$ to $15\%$ per day \cite{johnson2001role}, which leads to an estimated complete turnover in a period of nearly one week. However, when a new tissue has to be formed, e.g. to heal a wound, the rate of production is one or two order of magnitude faster \cite{chiquet1996regulation,dejana1990fibrinogen}.  It is also well known that the remodeling process is strongly affected by the stress and the strain the tissue undergoes, as clearly occurs for bones, teeth, and muscles \cite{kim2002gene,kjaer2004role,mao2004growth}. Hence, the relation between the rate of production/degradation of ECM constituents and the pressure felt by the cells is rather complicated.

Increased presence of ECM does not characterize only many tumors but was also observed in other pathologies like intima hyperplasia, cardiac, liver, and pulmonary fibrosis, asthma, colon cancer \cite{berk2007ecm,brewster1990myofibroblasts,iredale2007models,johnson2001role,liotta2001microenvironment,pinzani2000liver}. The alteration in the ECM composition can be due to several probably concurring reasons, including increased synthesis of ECM proteins, decreased activity of MMPs, upregulation of Tissue-specific Inhibitors of MetalloProteinases (TIMPs).

The interaction between ECM and cells is also attracting the attention of many research for other reasons. Indeed, on the one hand cells must adhere properly in order to survive, a phenomenon called \emph{anoikis}, as well as be anchored to the ECM to undergo mithosis. On the other hand, the interaction with the stroma has been argued to be one of the causes of tumor progression \cite{butcher2009tense,hautekeete1997hepatic,
kass2007mammary,liotta2001microenvironment,ruiter2002melanoma}.

In this paper we propose a general multiphase mathematical model able to describe the formation of fibrosis through either excessive production of ECM or underexpression of MMPs. The model is based on the frameworks deduced in \cite{chaplain2006mml,MR2471305}, taking also cell-ECM adhesion into account. In particular, ECM is regarded as a rigid scaffold while the cell populations (tumor and healthy cells) are assumed to behave similarly to elastic fluids. More realistic constitutive models, taking cell-cell adhesion into account and comparing theoretical and experimental results, can be found in \cite{ambrosi2008cam,preziosi2009evp}. For the sake of conciseness, we refrain from citing here all papers dealing with multiphase models of tumor growth, and refer to the recent reviews \cite{MR2253816,preziosi2009mmt,tracqui2009bmt} for more references. 

In more detail, Sect. \ref{sect:multiphase} derives and describes the model, which is then studied from the qualitative point of view in Sect. \ref{sect:spathomog}, having in mind the general dependence on the parameters stemming from biology. Existence, uniqueness, and continuous dependence of the solution on the initial data is proved in the spatially homogeneous case, and equilibrium configurations are discussed. These theoretical investigations reveal several interesting features of the model, for instance the fact that it predicts no other equilibria but the fully physiological and the fully pathological ones, featuring no tumor cells and no healthy cells, respectively. The physiological equilibrium turns out to be stable in the manifold with no tumor cells, but becomes unstable as soon as few tumor cells are present, which trigger the formation of fibrotic tissue.

\section{Multiphase modeling: general picture and particular cases}
\label{sect:multiphase}
In the multiphase modeling approach, tumors are regarded as a mixture of several interacting components whose main state variables are the volume ratios, i.e., their percent amounts within the mixture. With a view to providing a simplified, though still realistic, description of the system, we confine our attention to two cell populations: tumor cells, with volume ratio $\phi_T$, and healthy cells, with volume ratio $\phi_H$, moving within a remodeling extracellular matrix with volume ratio $\phi_M$. Clearly, $0\leq\phi_\alpha\leq 1$ for all $\alpha=T,\,H,\,M$.

\subsubsection*{Balance equations for the cellular matter}
\label{sect:celleq}
Following \cite{MR2471305}, we obtain the main governing equations for the cellular matter by joining the mass balance equation and the corresponding balance of linear momentum (with inertial effects neglected):
\begin{equation}
	\frac{\partial\phi_\alpha}{\partial t}-\nabla\cdot\left(\phi_\alpha\left(\frac{\phi_\alpha}{\phi}
		-\frac{\sigma_{\alpha M}}{\vert\nabla{(\phi\Sigma(\phi))}\vert}\right)^+\mathbb{K}_{\alpha M}
			\nabla{(\phi\Sigma(\phi))}\right)=\Gamma_\alpha\phi_\alpha,
	\label{eq:celleq}
\end{equation}
where $\phi:=\phi_T+\phi_H$, $\Gamma_\alpha$ is the duplication/death rate, and $\mathbb{K}_{\alpha M}$ the cell motility tensor within the matrix. Cells are regarded as elastic balloons forming an isotropic fluid, and are assumed to feature equal mechanical properties, hence their stress tensor is $\T=-\Sigma(\phi)\mathbb{I}$ for a pressure-like function $\Sigma$. In addition, the model accounts for the attachment/detachment of the cells to/from the matrix by means of a stress threshold $\sigma_{\alpha M}\geq 0$, which switches cell velocity on or off according to the magnitude of the actual stress sustained by cells in interaction with the matrix (see \cite{MR2471305} for further details). 

In the application to matrix remodeling and fibrosis, we consider that cells duplicate and die mainly on the basis of the amount of matrix present in the mixture. In general, also the availability of some nutrients plays a major role, but in the present context we assume that they are always abundantly supplied to the cells. Specifically, we set
\begin{equation}
	\Gamma_\alpha=\gamma_\alpha(\phi_M)H_{\epsilon_\alpha}(\psi_\alpha-\psi)-\delta_\alpha
		-\delta'_\alpha H_{\epsilon_M}(m_\alpha-\phi_M),
	\label{eq:Gamma}
\end{equation}
where $\psi:=\phi_T+\phi_H+\phi_M$ is the overall volume ratio occupied by cells and matrix, and $\gamma_\alpha(\cdot)$ is the net growth rate of the cell population $\alpha$, tempered by the free space rate $H_{\epsilon_\alpha}$. In particular, the $H_{\epsilon_\alpha}$'s are functions bounded between $0$ and $1$, which vanish on $(-\infty,\,0)$ and equal $1$ on $(\epsilon_\alpha,\,+\infty)$ (further analytical details in Sect. \ref{sect:spathomog}, Assumption \ref{hp:param}). Cell growth is inhibited when the amount of free space locally available is too small ($\psi\geq\psi_\alpha$) with respect to a threshold $\psi_\alpha\in[0,\,1]$. At the same time, either apoptosis or anoikis can trigger cell death at rates $\delta_\alpha,\,\delta_\alpha'>0$, respectively, the latter taking place when a too small amount of ECM ($\phi_M\leq m_\alpha$) with respect to a given threshold $m_\alpha\in[0,\,1]$ results in an insufficient number of possible adhesion sites.

Usually $\gamma_T(\cdot)=\gamma_H(\cdot)$, $\delta_T=\delta_H$, $m_T=m_H$, $\epsilon_T=\epsilon_H=\epsilon_M$. Instead, a difference between $\psi_T$ and $\psi_H$, with $\psi_T>\psi_H$, may be related to a smaller sensitivity to contact inhibition clues by tumor cells \cite{chaplain2006mml}. On the whole, we notice that it must be
\begin{equation}
	\Gamma_T(\phi_M,\,\psi)>\Gamma_H(\phi_M,\,\psi), \qquad \forall\,(\phi_M,\,\psi)\in[0,\,1]\times[0,\,1]
	\label{eq:Gamma-ineq}
\end{equation}
which holds if: (i) $\delta'_T<\delta'_H$ (smaller sensitivity to anoikis by tumor cells), (ii) $\epsilon_T>\epsilon_H$ (different speed for the switch mechanism, e.g. because of a different uncertainty in the response to mechanical stimuli), (iii) $m_T<m_H$ (higher capability from tumor cells to escape anoikis by surviving a greater lack of adhesion sites).

\subsubsection*{Matrix remodeling}
In general, the ECM is a quite complicated fibrous me\-dium. For the sake of simplicity, we model it as a rigid scaffold, which makes it unnecessary to detail its stress tensor because the internal stress is indeterminate due to the rigidity constraint. Under this assumption, the evolution equation for the volume ratio $\phi_M$ reads
\begin{equation}
	\frac{\partial\phi_M}{\partial t}=\Gamma_M,
	\label{eq:ecmeq}
\end{equation}
where $\Gamma_M$ is the source/sink of ECM accounting for remodeling and degradation due to the motion of the cells within the scaffold. Notice that, in general, $\phi_M$ depends on both time $t$ and space $x$, although the latter acts mainly as a parameter in the above differential equation.

Matrix is globally remodeled by cells and degraded by MMPs, whose concentration per unit volume is denoted by $e=e(t,\,x)$:
\begin{equation}
	\Gamma_M=\sum_{\alpha=T,\,H}\mu_\alpha(\phi_M)H_{\epsilon_M}(\psi_M-\psi)\phi_\alpha-\nu e\phi_M.
	\label{eq:GammaM-1}
\end{equation}
Here, $\mu_\alpha$ is a nonnegative, nonincreasing function (cf. Sect. \ref{sect:spathomog}, Assumption \ref{hp:param}) representing the net matrix production rate by the cell population $\alpha$ tempered by the free space rate $H_{\epsilon_M}$, and $\nu>0$ is the degradation rate by the enzymes. As usual, the latter are not regarded as a constituent of the mixture, but rather as massless macromolecules diffusing in the extracellular fluid according to a reaction-diffusion equation: $e_t=D\Delta{e}+\sum_{\alpha=T,\,H}\pi_\alpha\phi_\alpha-e/\tau$,
for net production rates $\pi_\alpha>0$ and enzyme half-life $\tau>0$. Actually, enzyme dynamics is much faster than that involving cell growth and death, hence it is possible to work under a quasi-stationary approximation. Furthermore enzyme action is usually very local \cite{barker2000cws}, so that also diffusion can be neglected and finally $e=\tau\sum_{\alpha=T,\,H}\pi_\alpha\phi_\alpha.$
Inserting this expression into Eq. \eqref{eq:GammaM-1} and defining $\nu_\alpha:=\nu\tau\pi_\alpha$ ultimately yields
\begin{equation*}
	\Gamma_M=\sum_{\alpha=T,\,H}\left(\mu_\alpha(\phi_M)H_{\epsilon_M}(\psi_M-\psi)
		-\nu_\alpha\phi_M\right)\phi_\alpha.
	\label{eq:GammaM-2}
\end{equation*}
The pathological cases possibly leading to fibrosis are either $\mu_T(\cdot)>\mu_H(\cdot)$ or $\nu_T<\nu_H$, which imply that tumor cells produce either more ECM or less MMPs than healthy cells, respectively.

\section{The spatially homogeneous problem}
\label{sect:spathomog}
The spatially homogeneous problem describes the evolution of the system under the main assumption of absence of spatial variation of the state variables $\phi_T$, $\phi_H$, $\phi_M$. This allows in particular to describe the equilibria, and the related basins of attraction, as functions of the parameters of the model.

In the sequel we will be concerned with the following initial value problem:
\begin{equation}
	\left\{
	\begin{array}{rcll}
	 	\dfrac{d\phi_\alpha}{dt} & = & \left[\gamma_\alpha(\phi_M)H_{\epsilon_\alpha}(\psi_\alpha-\psi)-\delta_\alpha
	 		-\delta_\alpha'H_{\epsilon_M}(m_\alpha-\phi_M)\right]\phi_\alpha, \quad \alpha=T,\,H \\[0.3cm]
		\dfrac{d\phi_M}{dt} & = & \displaystyle{\sum_{\alpha=T,\,H}}
			(\mu_\alpha(\phi_M)H_{\epsilon_M}(\psi_M-\psi)-\nu_\alpha\phi_M)\phi_\alpha \\[0.6cm]
		\phi_\alpha(0) & = & \phi_\alpha^0\in[0,\,1], \quad \alpha=T,\,H,\,M
	\end{array}
	\right.
	\label{eq:ode}
\end{equation}
over a time interval $(0,\,T]$, $T>0$. Some preliminary technical assumptions are in order:
\begin{assumption}
For $\alpha=T,\,H,\,M$ as appropriate, we assume $0\leq\psi_\alpha,\,m_\alpha\leq 1$, $\delta_\alpha,\,\delta_\alpha',\,\nu_\alpha>0$, and in addition that $\gamma_\alpha,\,\mu_\alpha:[0,\,1]\to\R_+$ are Lipschitz continuous, with $\gamma_\alpha$ nondecreasing, $\gamma_\alpha(0)=0$, and $\mu_\alpha$ nonincreasing.

Moreover, we assume that the functions $H_{\epsilon_\alpha}:\R\to [0,\,1]$ are Lipschitz continuous and vanishing on $(-\infty,\,0]$.
\label{hp:param}
\end{assumption}

The monotonicity of $\gamma_\alpha,\,\mu_\alpha$ is dictated by the fact that cell proliferation is fostered by the presence of ECM, whereas production of new matrix is progressively inhibited by the accumulation of other matrix. Owing to the properties recalled in Assumption \ref{hp:param}, $\gamma_\alpha$, $\mu_\alpha$, and $H_{\epsilon_\alpha}$ satisfy
\begin{gather}
	\gamma_\alpha(s)\leq\Lip{\gamma_\alpha}s, \ \gamma_\alpha(s)\leq\gamma_\alpha(1), \
		\mu_\alpha(1)\leq\mu_\alpha(s)\leq\mu_\alpha(0), \ \forall\,s\in[0,\,1],
	\label{eq:propgm}
	\\
	H_{\epsilon_\alpha}(s-\beta)\leq\Lip{H_{\epsilon_\alpha}}\vert s\vert, \quad
		\forall\,s\in\R,\,\beta\geq 0.
	\label{eq:propH}
\end{gather}
The functions $H_{\epsilon_\alpha}$ may be taken to be mollifications of the Heaviside function, for instance $H_{\epsilon_\alpha}(s)=0$ if $s<0$, $H_{\epsilon_\alpha}(s)=\epsilon_\alpha^{-1}s$ if $0\leq s\leq\epsilon_\alpha$, and $H_{\epsilon_\alpha}(s)=1$ if $s>\epsilon_\alpha$, or even smoother.

Let us introduce the space $V^d:=C([0,\,T];\,\R^d)$ of continuous functions $\u:[0,\,T]\to\R^d$, endowed with the $\infty$-norm $\|\u\|_\infty=\max_{t\in[0,\,T]}\|\u(t)\|_1$. In proving our results we will utilize $d=3$ and $d=4$.

\subsubsection*{Well-posedness}
We start by studying existence, uniqueness, and continuous dependence on the data of the solution $\bphi=(\phi_T,\,\phi_H,\,\phi_M)$ to problem \eqref{eq:ode}. We will then also discuss its regularity. Since the $\phi_\alpha$'s are volume ratios, we are interested in nonnegative solutions such that $\psi(t)\leq 1$ for all $t\geq 0$.
\begin{theorem}[Existence, uniqueness, and continuous dependence]
\label{theo:wellpos}
For each initial datum $\bphi^0\geq 0$ with $\|\bphi^0\|_1\leq 1$, problem \eqref{eq:ode} admits a unique nonnegative global solution $\bphi\in C([0,\,+\infty);\,\R^3)$ such that $\|\bphi\|_\infty\leq 1$. In addition, if $\bphi_1,\,\bphi_2$ are the solutions corresponding to initial data $\bphi_1^0,\,\bphi_2^0$, then for each $T>0$ there exists a constant $C=C(T)>0$ such that
$$ \|\bphi_2-\bphi_1\|_\infty\leq C(T)\|\bphi_2^0-\bphi_1^0\|_1 $$
in the interval $[0,\,T]$.
\end{theorem}
\begin{proof}
\begin{enumproof}
\item Let us introduce the function $\varphi(t):=1-\psi(t)$ (which, in mixture theory, identifies the free space available to be filled by some extracellular fluid) and consider the auxiliary problem given by the set of equations \eqref{eq:ode} plus $\varphi'=-\sum_\alpha\phi_\alpha'$, along with $\varphi^0:=\varphi(0)=1-\sum_\alpha\phi_\alpha^0$. Clearly, a triple $(\phi_T,\,\phi_H,\,\phi_M)$ is a solution to problem \eqref{eq:ode} if and only if the quadruple $(\phi_T,\,\phi_H,\,\phi_M,\,\varphi)$ is a solution to the auxiliary problem.

\item We put the auxiliary problem in compact form as
\begin{equation}
	\begin{cases}
		\dfrac{d\bPhi}{dt}=J[\bPhi], \qquad t>0 \\[0.2cm]
		\bPhi(0)=\bPhi^0,
	\end{cases}
	\label{eq:ode-compact}
\end{equation}
where $\bPhi=(\bphi,\,\varphi)$ and $J:V^4\to V^4$ is given componentwise by the right-hand sides of the differential equations in \eqref{eq:ode} plus $J_\varphi=-\sum_\alpha J_\alpha$.

Next we make the substitution $\bPhi(t)=\bPsi(t)e^{-\lambda t}$, where $\lambda>0$ will be properly selected. Due to the specific expression of $J$, the term $J[\bPsi(t)e^{-\lambda t}]$ can be given the form $I[\bPsi](t)e^{-\lambda t}$ for a suitable operator $I:V^4\to V^4$, which allows us to rewrite problem \eqref{eq:ode-compact} in terms of $\bPsi$ as
\begin{equation*}
	\begin{cases}
		\dfrac{d\bPsi}{dt}=I[\bPsi]+\lambda\bPsi, \qquad t>0 \\[0.2cm]
		\bPsi(0)=\bPhi^0
	\end{cases}
\end{equation*}
or, in mild form, as $\bPsi(t)=\bPhi^0+\lint_0^t\Bigl(I[\bPsi](\tau)+\lambda\bPsi(\tau)\Bigr)\,d\tau=:G[\bPsi](t)$.

\item Let us look for a mild solution in the following set of admissible functions:
$$ \A=\{\u\in V^4\,:\,\u(t)\geq 0,\ \|\u(t)\|_1=e^{\lambda t}\ \text{for all\ } t\in[0,\,T]\}. $$
Notice that $\bPsi\in\A$ amounts in particular to $\phi_\alpha(t),\,\varphi(t)\geq 0$ with $\sum_\alpha\phi_\alpha(t)+\varphi(t)=1$ for all $t\in[0,\,T]$, thus $\sum_\alpha\phi_\alpha(t)\leq 1$, which is what the saturation constraint requires on the volume ratios of the constituents of the mixture.

Any mild solution of $\bPsi(t)=G[\bPsi](t)$ is a fixed point of the operator $G$, therefore the task is to show that $G$ admits a unique fixed point in $\A$.

\item Owing to Assumption \ref{hp:param} and properties \eqref{eq:propgm}, \eqref{eq:propH}, if $\u(t)\geq 0$ then
\begin{align*}
	G_\alpha[\u](t) &\geq \phi_\alpha^0+(\lambda-C_\alpha)\lint_0^t u_\alpha(\tau)\,d\tau
		\quad (C_{T,H}=\delta_{T,H}+\delta'_{T,H},\ C_M=\nu_T+\nu_H), \\
	G_\varphi[\u](t) &\geq \varphi^0+\left(\lambda-\sum_{\alpha=T,\,H}(\gamma_\alpha(1)\Lip{H_{\epsilon_\alpha}}
		+\mu_\alpha(0)\Lip{H_{\epsilon_M}})\right)\lint_0^tu_\varphi(\tau)\,d\tau,
\end{align*}
hence we can choose $\lambda>0$ so large that $G[\u](t)\geq 0$ as well. If in addition $\|\u(t)\|_1=e^{\lambda t}$ then, using $I_\varphi=-\sum_\alpha I_\alpha$, we discover $\|G[\u](t)\|_1=e^{\lambda t}$. In conclusion, $\u\in\A$ implies $G[\u]\in\A$, i.e., $G$ maps $\A$ into itself.

\item Take now $\u,\,\v\in\A$ and observe that
$$ \|G[\u](t)-G[\v](t)\|_1\leq\lint_0^t\left(\|I[\u](\tau)-I[\v](\tau)\|_1+
	\lambda\|\u(\tau)-\v(\tau)\|_1\right)\,d\tau. $$
Lipschitz continuity of $\gamma_\alpha,\,\mu_\alpha,\,H_{\epsilon_\alpha}$ along with $H_{\epsilon_\alpha}(s)\leq 1$ and properties \eqref{eq:propgm}, \eqref{eq:propH} imply that there exists $C>0$, independent of $T$, such that $\vert I_\alpha[\u](t)-I_\alpha[\v](t)\vert\leq C\|\u(t)-\v(t)\|_1$ each $\alpha=T,\,H,\,M$. Since $I_\varphi=-\sum_\alpha I_\alpha$, an analogous relationship holds true also for $I_\varphi$, hence finally $\|G[\u]-G[\v]\|_\infty\leq T(C+\lambda)\|\u-\v\|_\infty$, which proves that $G$ is Lipschitz continuous on $\A$.

\item From the above calculations we see that we can choose $T>0$ so small that $G$ be a contraction on $\A$. Since $\A$ is a closed subset of $V^4$, Banach Fixed Point Theorem asserts that $G$ has a unique fixed point $\bPsi\in\A$. Therefore, the auxiliary problem \eqref{eq:ode-compact} admits a unique nonnegative local solution $\bPhi\in V^4$ such that $\|\bPhi(t)\|_1=1$. The three first components of $\bPhi$ form the unique nonnegative solution $\bphi\in V^3$ to problem \eqref{eq:ode} with $\|\bphi(t)\|_1\leq 1$.

Next, taking $\bphi(T)$ as new initial condition and observing that it matches all the hypotheses satisfied by $\bphi^0$, we uniquely prolong $\bphi$ over the time interval $[T,\,2T]$ in such a way that $\bphi(t)\geq 0$ and $\|\bphi(t)\|_1\leq 1$ all $t\in[0,\,2T]$. Proceeding inductively, we ultimately end up with a unique nonnegative global solution $\bphi\in C([0,\,+\infty);\,\R^3)$, for which the estimate $\|\bphi\|_\infty\leq 1$ easily follows from $\|\bphi(t)\|_1\leq 1$ all $t\geq 0$.

\item Let now $\bPsi_1,\,\bPsi_2\in\A$ be the two mild solutions corresponding to initial data $\bPhi_1^0,\,\bPhi_2^0$. Using the previous estimates we discover
$$ \|\bPsi_2(t)-\bPsi_1(t)\|_1\leq\|\bPhi_2^0-\bPhi_1^0\|_1+
	(C+\lambda)\lint_0^t\|\bPsi_2(\tau)-\bPsi_1(\tau)\|_1\,d\tau, $$
whence, invoking Gronwall's inequality,
$$ \|\bPsi_2(t)-\bPsi_1(t)\|_1\leq\left[1+(C+\lambda)te^{(C+\lambda)t}\right]\|\bPhi_2^0-\bPhi_1^0\|_1. $$
Returning to $\bPhi_1,\,\bPhi_2$ and observing that $\|\bPhi_2^0-\bPhi_1^0\|_1\leq 2\|\bphi_2^0-\bphi_1^0\|_1$ we finally get the desired estimate of continuous dependence, after taking the maximum of both sides for $t\in[0,\,T]$. \qedhere
\end{enumproof}
\end{proof}

\begin{theorem}[Regularity]
If the functions $\gamma_\alpha,\,\mu_\alpha,\,H_{\epsilon_\alpha}$ are of class $C^k$ on $[0,\,1]$ then the solution $\bphi$ is of class $C^{k+1}$ on $[0,\,+\infty)$.
\end{theorem}
\begin{proof}
According to Theorem \ref{theo:wellpos}, the $\phi_\alpha$'s are continuous on $[0,\,+\infty)$, therefore the right-hand sides of the differential equations in \eqref{eq:ode} define continuous functions on $[0,\,+\infty)$. It follows that the $\phi_\alpha'$'s are continuous as well, i.e., the solution $\bphi$ is actually $C^1$ on $[0,\,+\infty)$. If $\gamma_\alpha,\,\mu_\alpha,\,H_{\epsilon_\alpha}$ are of class $C^k$ then, by differentiating the ODEs in \eqref{eq:ode} $k$ times, this reasoning can be applied inductively to discover $\bphi\in C^{k+1}([0,\,+\infty);\,\R^3)$.
\end{proof}

\subsubsection*{Stability of the equilibrium configurations}
Next we study the equilibria of model \eqref{eq:ode}. It is immediately seen that $\phi_T=\phi_H=0$ gives rise to an equilibrium solution for any $\phi_M\in[0,\,1]$, corresponding to that all cells have died leaving some ECM. In order to investigate nontrivial equilibrium configurations, we proceed by considering first the two important sub-cases in which either $\phi_T=0$ or $\phi_H=0$ but $\phi_T,\,\phi_H$ do not vanish at the same time. The former will be referred to as the fully physiological case, the latter as the fully pathological one. In the following, $\phi_\alpha$ will be the nonzero volume ratio for either $\alpha=T$ or $\alpha=H$, meaning that the other one is identically zero. For the sake of simplicity, let us fix $\psi_M=1$ and examine the case $\phi_\alpha+\phi_M\leq 1-\eta$ for some arbitrarily small $\eta>0$, in such a way that, choosing $\epsilon_M<\eta$, we have the simplification $H_{\epsilon_M}(\psi_M-\phi_\alpha-\phi_M)=1$.

Suppose that the function $\mu_\alpha(s)-\nu_\alpha s$ has exactly one zero, say $s=M_\alpha$. Since, for physiological reasons, we further must have $\mu_\alpha(\psi_\alpha)<\nu_\alpha\psi_\alpha$, it follows $M_\alpha\in(0,\,\psi_\alpha)$. In this case, there is one nontrivial equilibrium given by
\begin{equation}
	\phi_M=M_\alpha, \qquad \phi_\alpha=\psi_\alpha-M_\alpha-
		H_{\epsilon_\alpha}^{-1}\left(\frac{\delta_\alpha+
			\delta'_\alpha H_{\epsilon_M}(m_\alpha-M_\alpha)}{\gamma_\alpha(M_\alpha)}\right),
	\label{eq:equilH}
\end{equation}
which is readily checked to be stable. Notice that the function $H_{\epsilon_\alpha}^{-1}(s)$ is well defined for $s\in(0,\,1)$.

The trivial equilibrium with also $\phi_\alpha=0$ can be reached basically in two situations. The first one is when $\phi_M$ is initially too small, so that the growth rate of the cells is lower than the apoptotic rate and anoikis occurs. This corresponds to initial conditions located in the lower-left corner of the phase portrait illustrated in Fig. \ref{fig:phaseport}, left. The equation $\phi_\alpha=\phi_\alpha(\phi_M)$ of the curve delimiting this basin of attraction in the phase space is obtained by integrating
\begin{equation}
	\frac{d\phi_\alpha}{d\phi_M}=\frac{\gamma_\alpha(\phi_M)H_{\epsilon_\alpha}(\psi_\alpha-\phi_\alpha-\phi_M)
		-\delta_\alpha-\delta'_\alpha H_{\epsilon_M}(m_\alpha-\phi_M)}{\mu_\alpha(\phi_M)-\nu_\alpha\phi_M}
	\label{eq:ode-phaseport}
\end{equation}
with the condition $\phi_\alpha(\phi_{M\alpha}^\star)=0$, $\phi_{M\alpha}^\star\in(0,\,1)$ being the smaller root of the equation $\Gamma_\alpha(\phi_M)=0$ with $\phi_T=\phi_H=0$, cf. Eq. \eqref{eq:Gamma}, characterized by $\Gamma_\alpha'(\phi_{M\alpha}^\star)>0$. This region does not exist if $\phi_{M\alpha}^\star<0$.

\begin{figure}[t]
\begin{center}
\includegraphics[width=0.3\textwidth,clip]{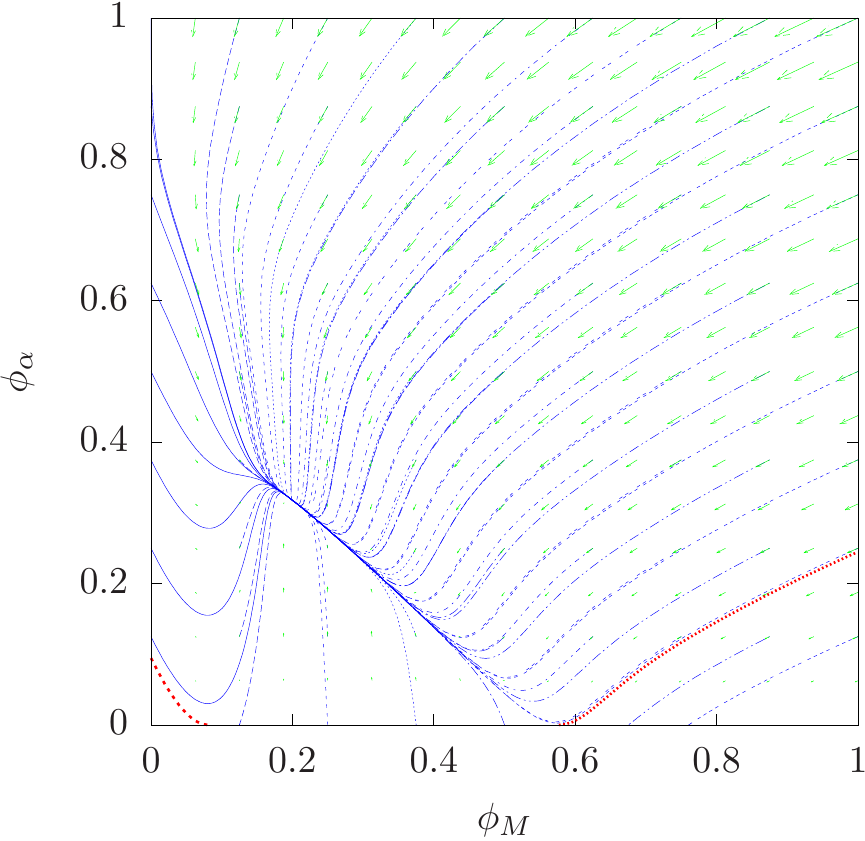} \qquad
\includegraphics[width=0.3\textwidth,clip]{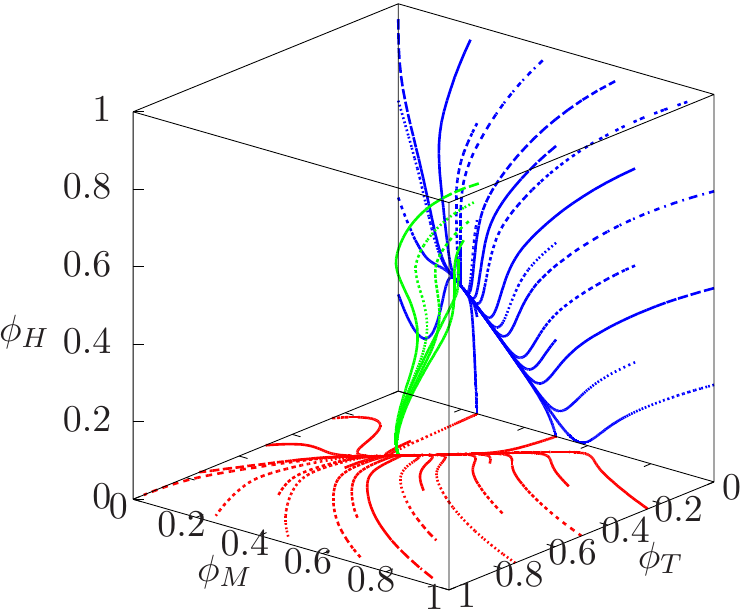}
\end{center}
\caption{Left: phase portrait in the fully physiological or pathological case. Right: phase portrait for the full model.}
\label{fig:phaseport}
\end{figure}

The second case is when $\phi_M$ is initially too large, namely the ECM is overly dense and cells are so compressed that the growth rate is again lower than the apoptotic rate because $H_{\epsilon_\alpha}(\psi_\alpha-\phi_\alpha-\phi_M)\approx 0$. This corresponds to initial conditions located in the lower-right corner of the phase portrait depicted in Fig. \ref{fig:phaseport}, left. The curve delimiting the basin of attraction is again obtained by integrating Eq. \eqref{eq:ode-phaseport}, now with the condition $\phi_\alpha(\phi_{M\alpha}^{\star\star})=0$, $\phi_{M\alpha}^{\star\star}\in(0,\,1)$ being the larger root of the equation $\Gamma_\alpha(\phi_M)=0$ with $\phi_T=\phi_H=0$, so that $\phi_{M\alpha}^\star\leq\phi_{M\alpha}^{\star\star}$. In this case $\Gamma_\alpha'(\phi_{M\alpha}^{\star\star})<0$. The region does not exist if $\phi_{M\alpha}^{\star\star}>1$.

In order to get the complete picture, we further have to investigate whether a nontrivial equilibrium with $\phi_T,\,\phi_H>0$ may exist. For this, we recall that the duplication/death rates $\Gamma_\alpha$ are constructed so as to match the biological requirement $\Gamma_T(\phi_M,\,\psi)>\Gamma_H(\phi_M,\,\psi)$ for all $(\phi_M,\,\psi)\in[0,\,1]\times[0,\,1]$, cf. Eq. \eqref{eq:Gamma-ineq}. As a consequence, we see that it is impossible for the right-hand sides of the two first equations of problem \eqref{eq:ode} to vanish simultaneously at an equilibrium point $(\phi_T,\,\phi_H,\,\phi_M)$ with $\phi_T,\,\phi_H\ne 0$, for this would imply that there exist $\phi_M\in[0,\,1]$ and $\psi\in(0,\,1]$ such that $\Gamma_\alpha(\phi_M,\,\psi)=0$ for both $\alpha=T,\,H$, which contradicts the above-mentioned Eq. \eqref{eq:Gamma-ineq}.

Hence, the only possible equilibria of the system are those arising in the fully physiological or pathological situation. In addition, condition \eqref{eq:Gamma-ineq} makes the nontrivial physiological equilibrium unstable and the nontrivial pathological one stable, as it can be realized from the three-dimensional phase portrait shown in Fig. \ref{fig:phaseport}, right.

\vskip0.3cm

\begin{figure}[t]
\begin{center}
\includegraphics[width=0.3\textwidth,clip]{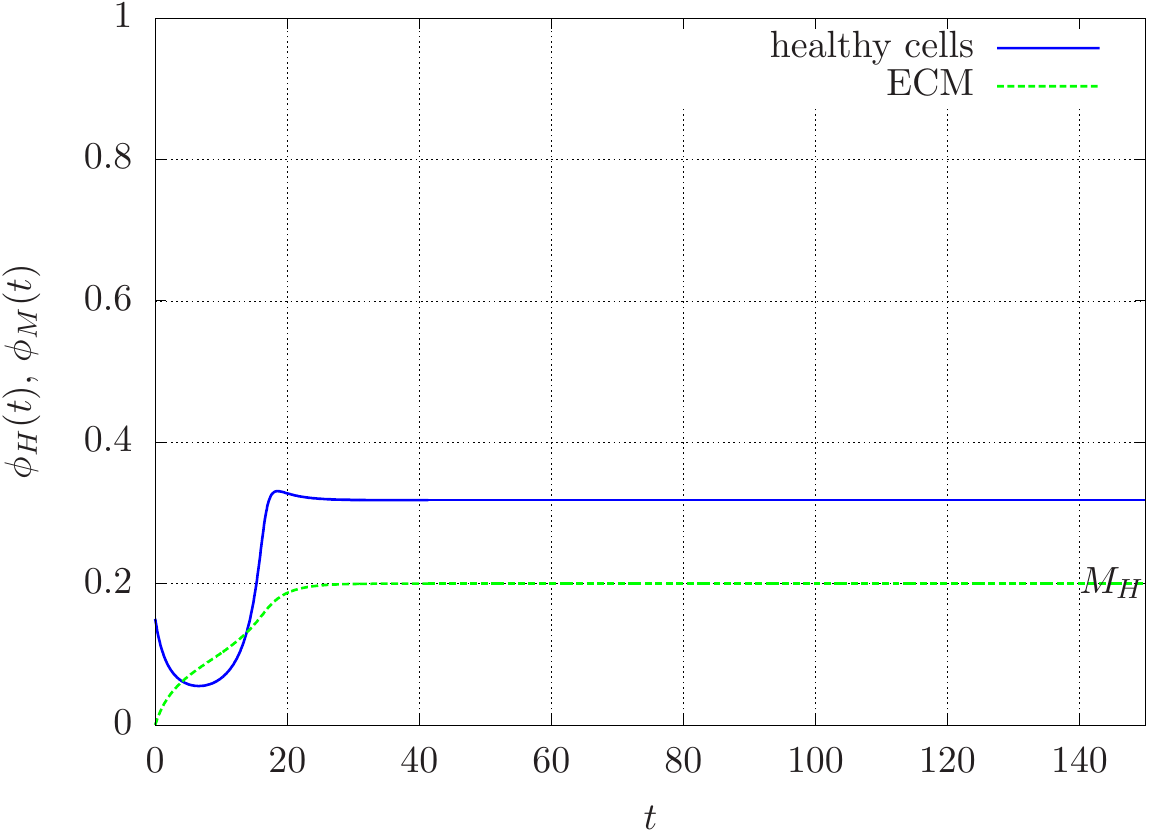} \qquad
\includegraphics[width=0.3\textwidth,clip]{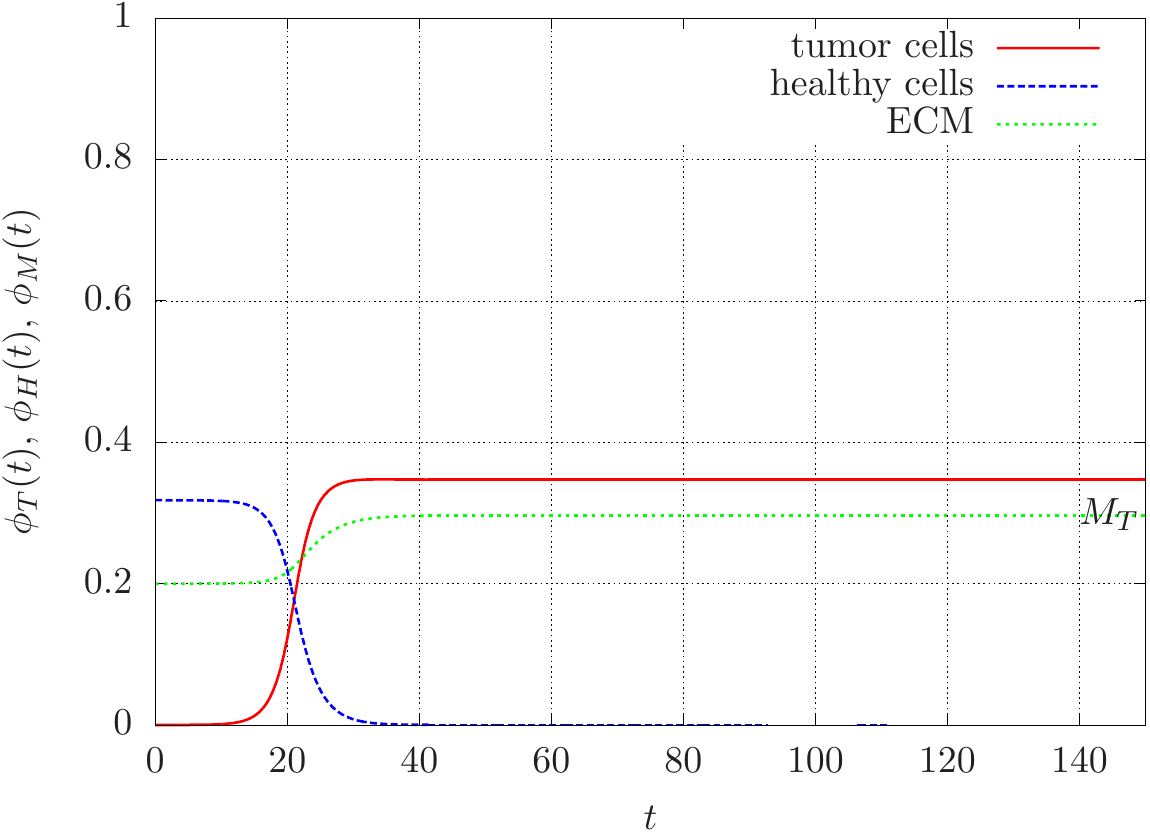}		
\caption{Left: formation of normal tissue in the physiological case. Right: formation of hyperplastic fibrotic tissue due to a small initial amount of tumor cells.}
\label{fig:evolution}
\end{center}
\end{figure}

Figure \ref{fig:evolution} (left) shows an example of a temporal evolution of the system giving rise to the formation of normal tissue in the fully physiological case. The initial death of healthy cells is due to anoikis, indeed cells are seeded in an environment completely deprived of ECM, that they have to build fast enough. The decrease stops as soon as the amount of ECM produced is such that $\gamma_H(\phi_M)H_{\epsilon_H}(\psi_H-\psi)\geq\delta_H+\delta'_HH_{\epsilon_M}(m_H-\phi_M)$, then the number of cells starts increasing, eventually leading to the stationary solution predicted by Eq. \eqref{eq:equilH} for $\alpha=H$. Conversely, if the initial amount of cells is insufficient to produce ECM rapidly enough then the entire population will die.

Figure \ref{fig:evolution} (right) gives instead an example of a complete temporal history ending with the formation of hyperplastic and fibrotic tissue. Despite the initial conditions $\phi_H^0,\,\phi_M^0$ coincide with the equilibrium values reached after the formation of normal tissue, the presence of a small amount of tumor cells ($\phi_T^0>0$) changes dramatically the outcome of the evolution, leading to a full depletion of healthy cells. This evolution can be duly compared with that shown in Fig. \ref{fig:spatinhomog}, which refers to the simulation of the full spatial and temporal model in one space dimension, cf. Eqs. \eqref{eq:celleq}, \eqref{eq:ecmeq}. Starting from the same initial conditions as in the spatially homogeneous case, the presence of a small amount of tumor cells at the beginning generates a traveling wave, which progressively depletes healthy cells and produces further fibrotic matrix while invading the normal tissue.

\begin{figure}[t]
\begin{center}
\includegraphics[width=0.32\textwidth,clip]{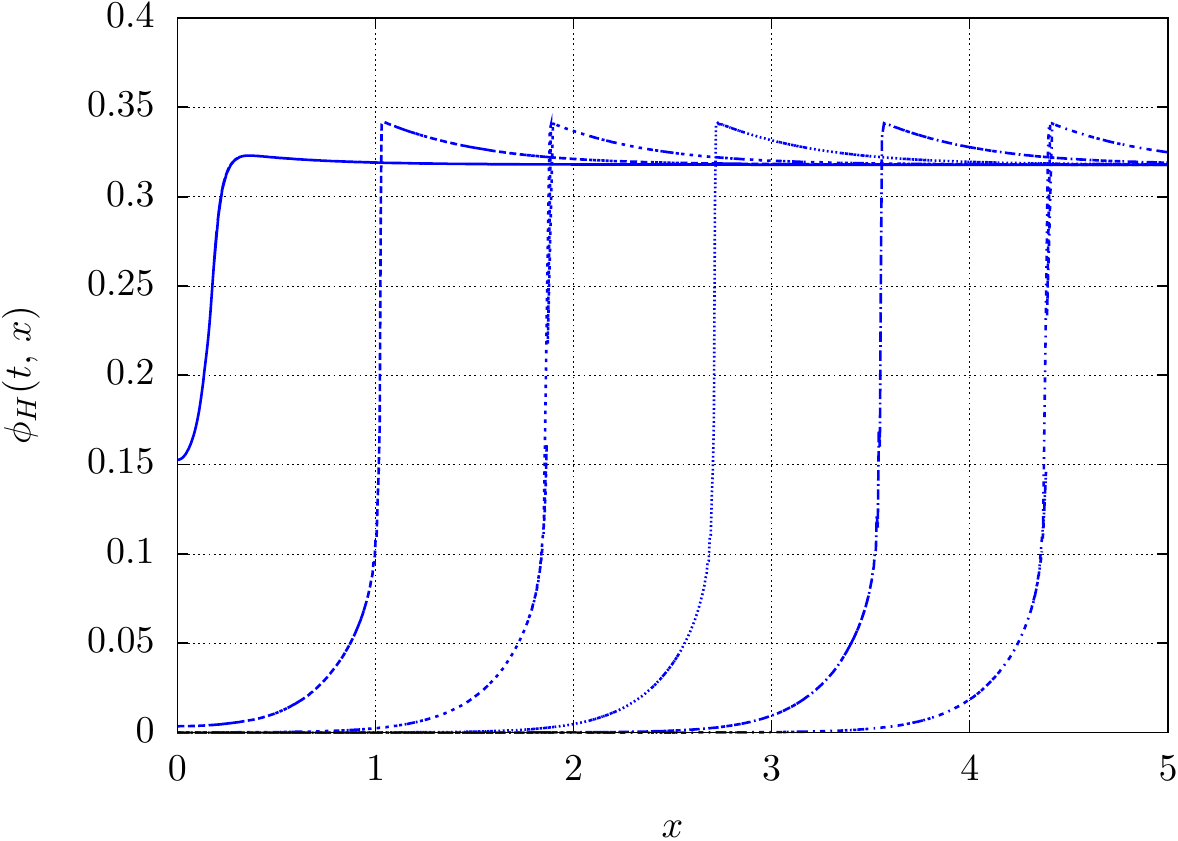}
\includegraphics[width=0.32\textwidth,clip]{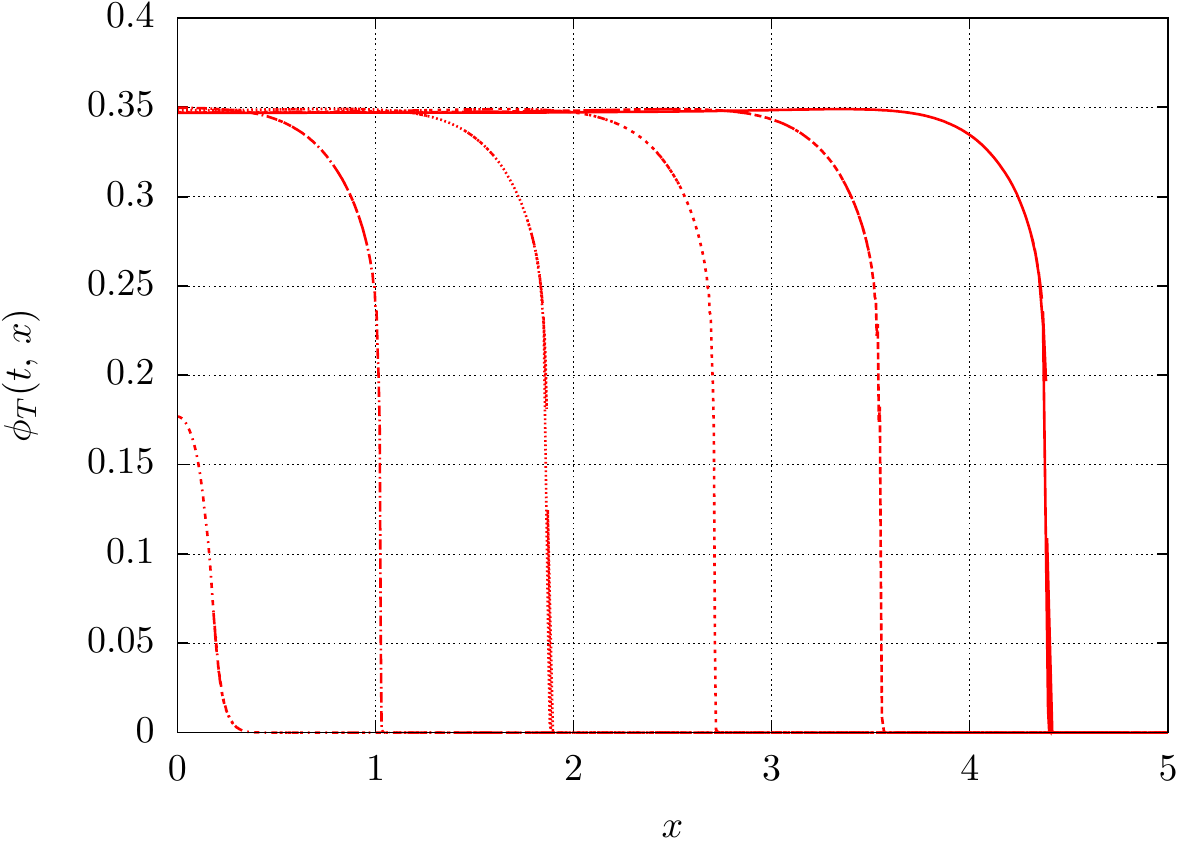}
\includegraphics[width=0.32\textwidth,clip]{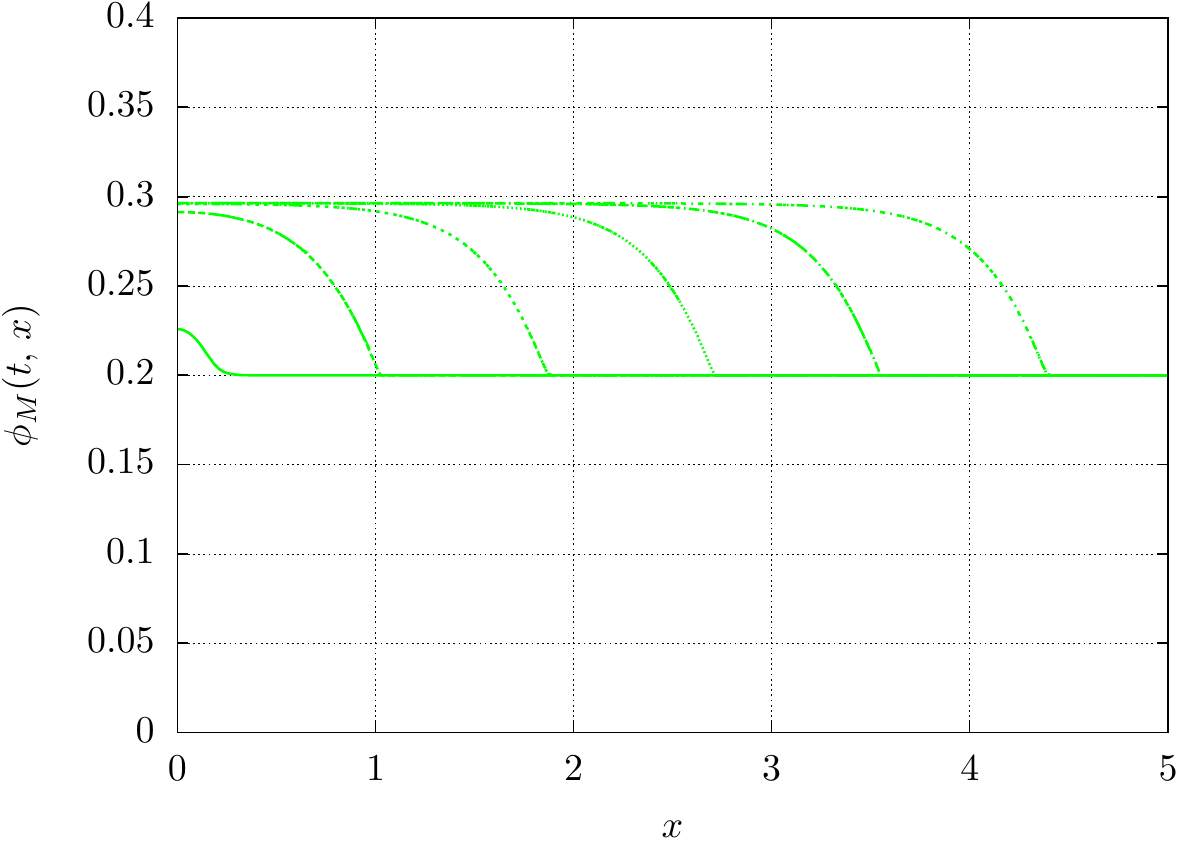}
\caption{Traveling wave solutions for $\phi_H$ (blue), $\phi_T$ (red), $\phi_M$ (green) in the full spatial and temporal evolution.}
\label{fig:spatinhomog}
\end{center}
\end{figure}


\begin{thebibliography}{10}
\expandafter\ifx\csname url\endcsname\relax
  \def\url#1{\texttt{#1}}\fi
\expandafter\ifx\csname urlprefix\endcsname\relax\def\urlprefix{URL }\fi
\expandafter\ifx\csname href\endcsname\relax
  \def\href#1#2{#2} \def\path#1{#1}\fi

\bibitem{paszek2005tensional}
M.~J. Paszek, N.~Zahir, K.~R. Johnson, J.~N. Lakins, G.~I. Rozenberg, A.~Gefen,
  C.~A. Reinhart-King, S.~S. Margulies, M.~Dembo, D.~Boettiger, D.~A. Hammer,
  V.~M. Weaver, Tensional homeostasis and the malignant phenotype, Cancer Cell
  8~(3) (2005) 241--254.

\bibitem{butcher2009tense}
D.~T. Butcher, T.~Alliston, V.~M. Weaver, A tense situation: forcing tumour
  progression, Nat. Rev. Cancer 9~(2) (2009) 108--122.

\bibitem{kass2007mammary}
L.~Kass, J.~T. Erler, M.~Dembo, V.~M. Weaver, Mammary epithelial cell:
  influence of extracellular matrix composition and organization during
  development and tumorigenesis, Int. J. Biochem. Cell B. 39~(11) (2007)
  1987--1994.

\bibitem{takeuchi1976variation}
J.~Takeuchi, M.~Sobue, E.~Sato, M.~Shamoto, K.~Miura, S.~Nakagaki, Variation in
  glycosaminoglycan components of breast tumors, Cancer Res. 36~(7 Pt 1) (1976)
  2133--2139.

\bibitem{zhang2003characteristics}
Y.~Zhang, S.~Nojima, H.~Nakayama, Y.~Jin, H.~Enza, Characteristics of normal
  stromal components and their correlation with cancer occurrence in human
  prostate, Oncol. Rep. 10~(1) (2003) 207--211.

\bibitem{johnson2001role}
P.~R. Johnson, Role of human airway smooth muscle in altered extracellular
  matrix production in asthma, Clin. Exp. Pharmacol. P. 28~(3) (2001) 233--236.

\bibitem{chiquet1996regulation}
M.~Chiquet, M.~Koch, M.~Matthisson, M.~Tannheimer, R.~Chiquet-Ehrismann,
  Regulation of extracellular matrix synthesis by mechanical stress, Biochem.
  Cell. Biol. 74~(6) (1996) 737--744.

\bibitem{dejana1990fibrinogen}
E.~Dejana, M.~G. Lampugnani, M.~Giorgi, M.~Gaboli, P.~C. Marchisio,
  Fibrinogen induces endothelial cell adhesion and spreading via the release
  of endogenous matrix proteins and the recruitment of more than one integrin
  receptor, Blood 75~(7) (1990) 1509--1517.

\bibitem{kim2002gene}
S.~G. Kim, T.~Akaike, T.~Sasagaw, Y.~Atomi, H.~Kurosawa, Gene expression of
  type i and type iii collagen by mechanical stretch in anterior cruciate
  ligament cells, Cell Struct. Funct. 27~(3) (2002) 139--144.

\bibitem{kjaer2004role}
M.~Kj{\ae}r, {Role of extracellular matrix in adaptation of tendon and skeletal
  muscle to mechanical loading}, Physiol. Rev. 84~(2) (2004) 649--698.

\bibitem{mao2004growth}
J.~J. Mao, H.-D. Nah, Growth and development: hereditary and mechanical
  modulations, Am J. Orthod. Dentofac. 125~(6) (2004) 676--689.

\bibitem{berk2007ecm}
B.~C. Berk, K.~Fujiwara, S.~Lehoux, {ECM} remodeling in hypertensive heart
  disease, J. Clin. Invest. 117~(3) (2007) 568--575.

\bibitem{brewster1990myofibroblasts}
C.~E. Brewster, P.~H. Howarth, R.~Djukanovic, J.~Wilson, S.~T. Holgate, W.~R.
  Roche, {Myofibroblasts and
  subepithelial fibrosis in bronchial asthma}, Am. J. Respir. Cell Mol. Biol.
  3~(5) (1990) 507--511.

\bibitem{iredale2007models}
J.~P. Iredale, Models of liver fibrosis: exploring the dynamic nature of
  inflammation and repair in a solid organ, J. Clin. Invest. 117~(3) (2007)
  539--548.

\bibitem{liotta2001microenvironment}
L.~A. Liotta, E.~C. Kohn, The microenvironment of the tumour-host interface,
  Nature 411 (2001) 375--379.

\bibitem{pinzani2000liver}
M.~Pinzani, Liver fibrosis, in: Springer seminars in immunopathology, Vol.~21,
  Springer, 2000, pp. 475--490.

\bibitem{hautekeete1997hepatic}
M.~L. Hautekeete, A.~Geerts, The hepatic stellate (ito) cell: its role in human
  liver disease, Virchows Arch. 430~(3) (1997) 195--207.

\bibitem{ruiter2002melanoma}
D.~Ruiter, T.~Bogenrieder, D.~Elder, M.~Herlyn, {Melanoma-stroma interactions:
  structural and functional aspects}, The Lancet Oncol. 3~(1) (2002) 35--43.

\bibitem{chaplain2006mml}
M.~A.~J. Chaplain, L.~Graziano, L.~Preziosi, Mathematical modelling of the loss
  of tissue compression responsiveness and its role in solid tumour
  development, Math. Med. Biol. 23~(3) (2006) 197--229.

\bibitem{MR2471305}
L.~Preziosi, A.~Tosin, Multiphase modelling of tumour growth and extracellular
  matrix interaction: mathematical tools and applications, J. Math. Biol.
  58~(4-5) (2009) 625--656.

\bibitem{ambrosi2008cam}
D.~Ambrosi, L.~Preziosi, Cell adhesion mechanisms and stress relaxation in the
  mechanics of tumours, Biomech. Model. Mechanobiol. 8~(5) (2008) 397--413.

\bibitem{preziosi2009evp}
L.~Preziosi, D.~Ambrosi, C.~Verdier, An elasto-visco-plastic model of cell
  aggregates, J. Theor. Biol.To appear.

\bibitem{MR2253816}
R.~P. Araujo, D.~L.~S. McElwain, A history of the study of solid tumour growth:
  the contribution of mathematical modelling, Bull. Math. Biol. 66~(5) (2004)
  1039--1091.

\bibitem{preziosi2009mmt}
L.~Preziosi, A.~Tosin, Multiphase and multiscale trends in cancer modelling,
  Math. Model. Nat. Phenom. 4~(3) (2009) 1--11.

\bibitem{tracqui2009bmt}
P.~Tracqui, Biophysical models of tumour growth, Rep. Prog. Phys. 72~(5) (2009)
  056701 (30pp).

\bibitem{barker2000cws}
N.~Barker, H.~Clevers, Catenins, {W}nt signaling and cancer, BioEssays 22~(11)
  (2000) 961--965.

\end{thebibliography}

\end{document}